\documentclass{llncs}
\usepackage{amssymb,amsmath,graphics,color}
\usepackage{graphicx}
\usepackage{multicol}
\usepackage{color}
\usepackage[noadjust]{cite}
\usepackage{enumitem}
\setlist{nolistsep}
\usepackage{hyperref}
\usepackage{comment}
\usepackage{pgf,tikz}
\usetikzlibrary{arrows,shapes}
\newcommand{\NP}{{\sf NP}}

\newcommand{\PSPACE}{{\sf PSPACE}}

\DeclareMathOperator{\dgr}{deg}

\date{}
\author{Carl Feghali, Matthew Johnson, Dani\"el Paulusma\thanks{Author supported by EPSRC (EP/K025090/1).}}

\institute{
School of Engineering and  Computing Sciences, Durham University,\\
Science Laboratories, South Road,
Durham DH1 3LE, United Kingdom
\texttt{\{carl.feghali,matthew.johnson2,daniel.paulusma\}@durham.ac.uk}
}

\title{A Reconfigurations Analogue \\ of Brooks' Theorem and its Consequences\thanks{An extended abstract of this paper appeared in the proceedings of MFCS 2014~\cite{DBLP:conf/mfcs/FeghaliJP14}.}}

\begin{document}
\maketitle

\begin{abstract}
Let $G$ be a simple undirected graph on $n$ vertices with maximum degree~$\Delta$.  Brooks' Theorem states that $G$ has a $\Delta$-colouring unless~$G$ is a complete graph, or a cycle with an odd number of vertices.  To recolour $G$ is to obtain a new proper colouring by changing the colour of one vertex.  We show an analogue of Brooks' Theorem by proving
that  from any $k$-colouring, $k>\Delta$, a $\Delta$-colouring of $G$ can be obtained by a sequence of $O(n^2)$ recolourings using only the original $k$ colours unless
\begin{itemize}
\item $G$ is a complete graph or a cycle with an odd number of vertices, or
\item $k=\Delta+1$, $G$ is $\Delta$-regular and, for each vertex $v$ in $G$, no two neighbours of $v$ are coloured alike.
\end{itemize}
We use this result to study the reconfiguration graph $R_k(G)$ of the  $k$-colourings of  $G$.  The vertex set of $R_k(G)$ is the set of all possible $k$-colourings of $G$ and two colourings are adjacent if they differ on exactly one vertex.
We prove that for $\Delta\geq 3$, $R_{\Delta+1}(G)$ consists of isolated vertices and at most one further component which has diameter $O(n^2)$.
This result enables us to complete both a structural classification and an algorithmic classification for reconfigurations of colourings of graphs of bounded maximum degree.
\end{abstract}

\section{Introduction}

Let $G=(V, E)$ denote a simple undirected graph and let $k$ be a positive integer.
A \emph{$k$-colouring} of $G$ is a function $\gamma: V \rightarrow \{1, 2, \ldots, k\}$ such that if $uv \in E$, $\gamma(u) \neq \gamma(v)$. 
The {\it degree} $\dgr(v)$ of a vertex $v\in V$ is the number of edges incident with $v$, 
or equivalently how many neighbours it has in $G$; we let $\Delta$ denote the maximum degree of $G$. The celebrated theorem of Brooks \cite{brooks} states that a graph $G$ has a $\Delta$-colouring unless $G$ is the complete graph on $\Delta+1$ vertices or a cycle with an odd number $n$ of vertices. Our goal is to translate Brooks' theorem to the setting of 
\emph{reconfiguration graphs}. 

Given a search problem one can define a corresponding reconfiguration graph as follows: vertices correspond to solutions and edges join solutions that are, in some sense, ``close'' to one another.
As this definition suggests, for a given search problem there might be more than one way to define an edge relation of the reconfiguration graph. Reconfiguration graphs have not only been studied for colouring but also for many other problems including boolean satisfiability~\cite{GKMP09,MTY10,Schw14},
clique and vertex cover~\cite{IDHPSUU10}, independent set~\cite{Bo14,BKW14,KMM12},  list edge colouring~\cite{IKD09,IKZ11}, {$L(2,1)$-labeling}~\cite{IKOZ12}, {shortest path}~\cite{Bo10,Bo13}, and {subset sum}~\cite{ID11}; see also a recent survey~\cite{He13}. 
Typical questions are: is the reconfiguration graph connected; if so what is its diameter; if not what is the diameter of its 
(connected) components; and how difficult is it to decide whether there is a path between a pair of given solutions? 
Recent work has included looking at finding the shortest path in the reconfiguration graph between given solutions~\cite{KMM11}, and studying the fixed-parameter-tractability of this problem~\cite{BM14, JKKPP14, MNR14, MNRSS13}.

For the colouring problem, the following definition of the reconfiguration graph is the
most natural. The \emph{$k$-colouring reconfiguration graph} of  $G$, denoted $R_k(G)$, has as its vertex set all possible $k$-colourings of $G$, and two $k$-colourings $\gamma_1$ and $\gamma_2$ are joined by an edge if, for some vertex $u \in V$, $\gamma_1(u) \neq \gamma_2(u)$, and, for all $v \in V \setminus \{u\}$, $\gamma_1(v) = \gamma_2(v)$; that is, if $\gamma_1$ and $\gamma_2$ \emph{disagree} on exactly one vertex.    

As mentioned, besides determining a bound on the diameter of the reconfiguration graph or of its components,
another common aim in this area is to find a path between a
given pair of colourings $\alpha$ and $\beta$ in a reconfiguration graph.  
This leads to the following decision problem (where $k$ denotes a fixed integer, that is, $k$ is not part of the input):

 \medskip
\noindent
\textsc{$k$-Colour Path}\\[2pt]
\begin{tabular}{p{1.7cm}p{10cm}}
\textit{Instance}\,:& A graph $G=(V,E)$ and two 
$k$-colourings $\alpha$ and $\beta$.\\
\textit{Question}\,:&Is there a path in $R_k(G)$ between $\alpha$ and 
$\beta$? 
\end{tabular}

\medskip
\noindent
Note that an equivalent formulation of this problem is whether there exists a sequence of colourings $\gamma_0, \gamma_1, \ldots, \gamma_t$ with $\alpha=\gamma_0$, $\beta=\gamma_t$ such that adjacent colourings disagree on a single vertex.  
We call this a \emph{recolouring} sequence.  If, for $1 \leq i \leq t$, $v_i$ is the vertex on which $\gamma_i$ and $\gamma_{i-1}$ disagree, then we can think of $\beta$ as being obtained from $\alpha$ by recolouring the vertices $v_1, \ldots, v_t$ in order.  Therefore, rather than explicitly considering the reconfiguration graph, one could seek to find a recolouring sequence of $G$; that is, to describe a sequence of vertices and to say which colour each vertex should be recoloured (while avoiding that two adjacent vertices are coloured alike).

\subsection{Existing Results}

The study of reconfiguration graphs of colourings was initiated by Cereda, van den Heuvel and Johnson~\cite{CHJ06, CHJ06a} who proved some initial results on the connectivity of reconfiguration graphs.
The {\sc $k$-Colour Path} problem was shown to be solvable in  time $O(n^2)$ for $k=3$ by Cereceda, van den Heuvel and Johnson~\cite{CHJ06b}; they also proved that the diameter of any component of  the reconfiguration graph~$R_3(G)$ of a 3-colourable graph $G$ is $O(n^2)$.  In contrast, Bonsma and Cerecada~\cite{BC09} proved that this problem
is \PSPACE-complete for all $k\geq 4$ even for bipartite graphs (and for bipartite planar graphs for $4 \leq k \leq 6$), and examples of reconfiguration graphs with components of superpolynomial diameter were given in all these cases.  

Bonamy et al.~\cite{BJLPP14} showed that reconfiguration graphs of $k$-colourings of chordal graphs are connected with diameter $O(n^2)$ whenever~$k$ is more than the size of the largest clique  (and they gave an infinite class of chordal graphs  whose reconfiguration graphs have diameter $\Omega(n^2)$).  Bonamy and Bousquet~\cite{BB13}  generalized this result by showing that if $k$ is at least two greater than the treewidth $tw(G)$ then, again, $R_k(G)$ is connected with diameter $O(n^2)$; note that if $k=tw(G)+1$, then $R_k(G)$ might not be connected since, for example,  $G$ might be a complete graph on $tw(G)+1$ vertices and then $R_k(G)$ contains no edges.  

Bousquet and Perarnau~\cite{BP14} considered sparse graphs. They proved that, for all $d\geq 0$, $k\geq d$ and $\epsilon>0$, the reconfiguration graph $R_k(G)$ of every $(d+1)$-colourable graph $G$ has a polynomial diameter provided that the maximum average degree of $G$ is at most $d-\epsilon$.  We will mention other related results later.

\subsection{Our Results}\label{s-ourresults}

We study reconfigurations of colourings for graphs of bounded maximum degree. Let $K_n$ and $C_n$ denote the complete graph and cycle on $n$ vertices, respectively. Recall that Brooks' Theorem states that every graph $G$ has a $\Delta$-colouring unless $G$ is isomorpic to $K_{\Delta+1}$ or 
$C_n$ for odd $n$. Our first result is an analogue of this theorem for reconfiguration graphs, that is, we answer the question:
given a $k$-colouring~$\gamma$ of $G$, $k \geq \Delta+1$, is there a path from $\gamma$ to a $\Delta$-colouring in~$R_k(G)$? 
(Note that, for any two integers $k$ and $k'$ with $k\geq k'$, every $k'$-colouring of $G$ corresponds to a vertex of $R_k(G)$ since a $k'$-colouring is a $k$-colouring in which not all colours are used.)

In order to state our results we require two definitions.  A $k$-colouring $\gamma$ of a graph is \emph{frozen} if, for every vertex $v$, every colour except $\gamma(v)$ is 
used on the neighbours of $v$.  Notice that a frozen colouring is an isolated vertex in $R_k(G)$.  The length of a shortest path between colourings $\alpha$ and $\beta$ in~$R_k(G)$ is denoted by $d_k(\alpha, \beta)$.  We state our result for connected graphs as disconnected graphs can be considered component-wise.

\begin{theorem} \label{t-delta+1}
Let $G$ be a connected graph on $n$ vertices with maximum degree~$\Delta \geq 1$, and let $k \geq \Delta+1$.    Let $\alpha$ be a $k$-colouring of $G$.  If $\alpha$ is not frozen and $G$ is not $K_{\Delta+1}$ or, if $n$ is odd, $C_n$, then there exists a $\Delta$-colouring $\gamma$ of $G$ such that~$d_k(\alpha, \gamma)$ is $O(n^2)$.  
\end{theorem}
If all vertices of a graph have degree $d$, it is called {\it $d$-regular} (or just {\it regular}). 

Note that $\alpha$ can only be frozen if $k=\Delta+1$, and only if $G$ is $\Delta$-regular.  Let us briefly note that such colourings do exist: for example a 3-colouring of~$C_6$ in which each colour appears exactly twice on vertices at distance 3, or a 4-colouring of the cube in which diagonally opposite vertices are coloured alike. In fact, as we will see, the case $k=\Delta+1$ is the only cause of difficulty in the proof of our first result, which can be found in Section~\ref{s-main1}.

Using Theorem~\ref{t-delta+1} we can, with the aid of a result of  Matamala~\cite{matamala} on 
partitioning graphs into two degenerate graphs, give a characterization of $R_{\Delta+1}(G)$ 
for $\Delta\geq 3$, 
which is our next result and is proved in Section~\ref{s-main2}.

\begin{theorem} \label{t-main}
Let $G$ be a connected graph on $n$ vertices  with maximum degree~$\Delta \geq 3$.    Let $\alpha$ and $\beta$ be $(\Delta+1)$-colourings of $G$.  If  $\alpha$ and $\beta$ are not frozen colourings, then $d_{\Delta+1}(\alpha, \beta)$ is $O(n^2)$. 
\end{theorem}
Theorem~\ref{t-main} implies that $R_{\Delta+1}(G)$ contains a number of isolated vertices (representing frozen colourings) plus, possibly, one further component. 
We observe that the requirement that $\Delta \geq 3$ is necessary since, for example, $R_3(C_n)$, $n$ odd, has more than one component consisting of at least two vertices~\cite{CHJ06,CHJ06a}.

It is possible that the number of isolated vertices is zero; that is, there are no frozen $(\Delta+1)$-colourings.
For example,  suppose that $G$ is a connected regular graph on $n\not\equiv 0 \mod{(\Delta+1)}$ vertices with maximum degree $\Delta \geq 3$, and let $V_1, V_2, \dots, V_{\Delta+1}$ be the colour classes of a frozen $(\Delta+1)$-colouring~$\gamma$.  Then, by defintion, for all $i,j$, $i\neq j$, each $v \in V_i$ has a neighbour in~$V_j$ and cannot have more than one neighbour in $V_j$, as it has $\Delta$ neighbours in total. Hence, $|V_1| = \dots = |V_{\Delta+1}|$ and thus $n \equiv 0 \mod{(\Delta+1)}$, 
a contradiction.  We note that connected $\Delta$-regular graphs on $n$ vertices can always be found (unless $n$ and $\Delta$ are both odd): for example, take $n$ vertices arranged on a circle and join each to the nearest $\lfloor \Delta/2 \rfloor$ vertices on either side and also, if $\Delta$ is odd, to the diametrically opposite vertex.
 
It is also possible that there are only isolated vertices. Consider $R_4(K_4)$ for instance; and Brooks' Theorem tells us that complete graphs are the only graphs for which $R_{\Delta+1}(G)$ is edgeless, since other graphs have colourings in which only $\Delta$ colours are used and by recolouring any vertex with the unused colour we find a neighbouring colouring.

\subsection{Two Classification Results}\label{s-consequences}

Theorem~\ref{t-main} enables us to complete both a structural classification and an algorithmic classification for reconfigurations of colourings of graphs with bounded maximum degree. 
In order to explain this we need to introduce some more terminology.

Thoughout the paper we let $n$ denote the number of vertices of a graph. 
We distinguish four types of classes of $k$-colourable graphs for our structural classification. As we will see, these four types also roughly correspond to four types of complexity results.
We say that a graph class 
${\cal G}$ of $k$-colourable graphs is of\\[-6pt] 
\begin{itemize}
\item {\bf type 1} if, for all $G\in {\cal G}$, $R_k(G)$ is connected and has diameter $O(n^2)$;
\item {\bf type 2} if, for all $G\in {\cal G}$, each component of $R_k(G)$ has diameter $O(n^2)$ and~$R_k(G)$ has at most one component that is not an isolated vertex;
\item {\bf type 3} if, for all $G\in {\cal G}$, each component of $R_k(G)$ has diameter $O(n^2)$;
\item {\bf type 4} if, for all $G\in {\cal G}$, $R_k(G)$ is disconnected and has at least one component with a superpolynomial diameter.
\end{itemize}

\medskip
\noindent
Note that every graph class of type~1 is of type~2 and that every graph class of type~2 is of type 3. 
At this point the reader may wonder whether there exists a class of graphs whose reconfiguration graph of $k$-colourings is connected but does not have an (at most) quadratic diameter.  This is still an open problem~(see, for example,~\cite{BJLPP14}). The structural classification presented in Theorem~\ref{t-mainstructural} below implies that if such a graph class exists then it contains graphs whose maximum degree is unbounded.
 
For integers $k\geq 1$ and $\Delta\geq 0$,
let ${\cal G}_k^\Delta$ be the class of  connected $k$-colourable graphs with maximum degree 
at most~$\Delta$. 
Note that ${\cal G}_1^\Delta=\emptyset$ if $\Delta\geq 1$ and that ${\cal G}_k^i\subseteq {\cal G}_k^j$ for any two integers $i$ and $j$ with $i\leq j$. 

We are now ready to formally state the consequences of our
earlier results. Theorem~\ref{t-mainstructural} classifies the connectivity and the diameter of the reconfiguration graph of a graph of bounded degree in terms of the four types defined above.
Theorem~\ref{t-mainalgorithmic} completely determines the computational complexity of the {\sc $k$-Colour Path} problem restricted to graphs of bounded degree. We obtain these two classification results by combining Theorem~\ref{t-main} with a number of results 
from the literature.

\begin{theorem}\label{t-mainstructural}
Let $k\geq 1$ and $\Delta\geq 0$ be integers. Then:\\[-8pt]
\begin{itemize}
\item [(i)] ${\cal G}_k^\Delta$ is of type~1 if 
\begin{itemize}
\item [$\bullet$] $k=1$ and $\Delta=0$
\item [$\bullet$] $k\geq 2$ and $\Delta\leq k-2$.\\[-8pt]
\end{itemize}
\item [(ii)]  ${\cal G}_k^\Delta$ is of type~2 if
\begin{itemize}
\item [$\bullet$] $k=2$ and $\Delta\geq 1$
\item [$\bullet$] $k\geq 4$ and $\Delta=k-1$.\\[-8pt]
\end{itemize}
\item [(iii)] ${\cal G}_k^\Delta$ is of type~3 if
\begin{itemize}
\item [$\bullet$] $k=3$ and $\Delta\geq 2$.\\[-8pt]
\end{itemize}
\item [(iv)] ${\cal G}_k^\Delta$ contains a subclass of type~4 
if
\begin{itemize}
\item [$\bullet$] $k\geq 4$ and $\Delta\geq k$.
\end{itemize}
\end{itemize}
\end{theorem}

\begin{proof}
We prove each of the four statements separately.
\begin{itemize}
\item[(i)] The case $k=1$ and $\Delta=0$ is trivial. 
The case $k\geq 2$ and $\Delta\leq k-2$ has been shown by 
Dyer, Flaxman, Frieze and Vigoda~\cite{DFFV06}; see also~\cite{BC09,CHJ06,luisthesis} for a proof.
\item[(ii)] The case $k=2$ and $\Delta\geq 1$ follows from the fact that ${\cal G}_2^\Delta$ consists of connected bipartite graphs. 
Hence, the corresponding reconfiguration graphs are either edgeless or isomorphic to a single edge (if the bipartite graph consists of a single vertex). 
The case $k\geq 4$ and $\Delta=k-1$ follows from Theorem~\ref{t-main}.
\item[(iii)]  This case has been proven by  Cereceda, van den Heuvel and Johnson~\cite{CHJ06b}.
\item[(iv)]  Let $k\geq 4$ and $\Delta\geq k$.  Bonsma and Cereceda~\cite{BC09} constructed an infinite family of $k$-colourable graphs whose reconfiguration graphs 
have components of superpolynomial diameter. It can be observed that these graphs belong to ${\cal G}_k^k$, and hence, to ${\cal G}_k^\Delta$ for all $\Delta\geq k$.
\qed
\end{itemize}
\end{proof}

\pagebreak

\begin{theorem}\label{t-mainalgorithmic}
Let $k\geq 1$ and $\Delta\geq 0$
be integers. Then 
{\sc $k$-Colour Path} restricted to ${\cal G}^\Delta_k$ is
\begin{itemize}
\item [(i)] solvable in $O(1)$ time if 
\begin{itemize}
\item [$\bullet$] $k\leq 2$
\item [$\bullet$] $k\geq 3$ and $\Delta\leq k-2$;\\[-8pt]
\end{itemize}
\item [(ii)] solvable in $O(n)$ time if 
\begin{itemize}
\item [$\bullet$] $k\geq 3$ and $\Delta=k-1$;\\[-8pt]
\end{itemize}
\item [(iii)] solvable in $O(n^2)$ time if
\begin{itemize}
\item [$\bullet$] $k=3$ and $\Delta\geq 3$;\\[-8pt]
\end{itemize}
\item [(iv)] \PSPACE-complete if
\begin{itemize}
\item [$\bullet$]  $k\geq 4$ and $\Delta\geq k$.
\end{itemize}
\end{itemize}
\end{theorem}

\begin{proof}
We prove each of the four statements separately.
\begin{itemize}
\item[(i)] This case follows from Theorem~\ref{t-mainstructural}~(i) (the answer is always yes) unless $k=2$ and $\Delta\geq k-1=1$. 
Recall from the proof of Theorem~\ref{t-mainstructural}~(ii) that in the latter case 
the reconfiguration graph is either edgeless or isomorphic to an edge. The answer is always no in the first case and yes in the second case.
\item[(ii)]
If $k=3$ and so $\Delta=2$, then $G$ is either a path or a cycle.  We know {\sc $k$-Colour Path} always has the answer yes for paths~\cite{luisthesis}, and can be decided for cycles by a single traversal of the edges~\cite{CHJ06b}.
Now let $k\geq 4$. By Theorem~\ref{t-mainstructural}~(ii), it is necessary in this case only to check for each vertex~$v$ in the input graph $G$, for each of the two given $k$-colourings $\alpha$ and $\beta$, whether $v$ and its neighbours use every colour in $\{1,2, \ldots, \Delta+1\}$.  If they do not, neither colouring is frozen, so there is a path between them.
\item[(iii)] This follows from~\cite{CHJ06b} for the superclass consisting of all 3-colourable graphs.
\item[(iv)] This follows from the aforementioned result of Bonsma and Cereceda~\cite{BC09} as
from their proof it  can be seen that the problem is \PSPACE-complete  for ${\cal G}_k^k$, and thus for ${\cal G}_k^\Delta$ for all $\Delta\geq k$.\qed
\end{itemize}
\end{proof}

\subsection{Further Work and Open Problems}\label{s-future}

We already mentioned the open problem on the existence of a class of graphs whose reconfiguration graph of $k$-colourings is connected but does not have an (at most) quadratic diameter.
We recall another open problem 
from the literarure which is on degenerate graphs and on which we can report some partial progress due to our new results.
A graph $G$ is {\it $k$-degenerate} if every induced subgraph of $G$ has a vertex with degree at most $k$.  Note that any graph is $\Delta$-degenerate.
Cereceda~\cite{luisthesis} made the following conjecture.

\begin{conjecture}\label{c-luis}
For any pair of integers $d,k$ with $k\geq d+2$, the reconfiguration graph $R_k(G)$ of a $d$-degenerate graph $G$ has diameter $O(n^2)$.
\end{conjecture}

It turns out that proving (or disproving) this conjecture is a very challenging problem even for $d=2$ and $k=4$. Using Theorem~\ref{t-main} we can solve one more case, as shown
in the next theorem which summarizes our current knowledge.

\begin{theorem}\label{t-conjecture}
Let $d\geq 0$ and  $k\geq d+2$, and let $G$ be a $d$-degenerate connected graph. Then $R_k(G)$ has diameter $O(n^2)$ if\\[-10pt]
\begin{itemize}
\item[(i)] $d=0$
\item[(ii)] $d=1$
\item[(iii)] $d=\Delta-1$
\item[(iv)] $d\geq \Delta$.
\end{itemize}
\end{theorem}

\begin{proof}
We prove each of the four statements separately.
\begin{itemize}
\item[(i)] This case is trivial.
\item[(ii)] Cereceda~\cite{luisthesis} proved that for any two integers $d$ and $k$ with $k\geq 2d+1$, the reconfiguration graph $R_k(G)$ of any graph $G$ has diameter~$O(n^2)$. Taking $d=1$ proves the case.
As an aside, Bousquet and Perarnau~\cite{BP14} proved that for any two integers $d$ and $k$ with $k\geq 2d+2$, the reconfiguration graph $R_k(G)$ of any graph~$G$ has 
diameter~$O(n)$.
\item[(iii)] If $k=d+2=\Delta+1$ then we can apply  Theorem~\ref{t-main} after observing that a $(\Delta - 1)$-degenerate graph has a vertex with at most $\Delta - 1$ neighbours, 
so no $k$-colouring $\alpha$ is frozen. If $k\geq d+3=\Delta+2$ then we apply Theorem~\ref{t-mainstructural}~(i).
\item[(iv)] This case follows from Theorem~\ref{t-mainstructural}~(i).
\qed
\end{itemize}
\end{proof}

Another direction for future work is to consider the problem of {\it finding} a path or a shortest path in the reconfiguration graph $R_k(G)$ between two given $k$-colourings $\alpha$ and $\beta$ of a graph $G$ of maximum degree~$\Delta$. 
For $k\geq 4$ and $\Delta\geq k$ this problem is \PSPACE-hard due to Theorem~\ref{t-mainalgorithmic}~(iv). 
However, for $1\leq k\leq 3$ or $0\leq \Delta\leq k-1$, this problem is not solved in statements~(i)--(iii) of Theorem~\ref{t-mainalgorithmic}, which correspond to exactly those cases for which {\sc $k$-Colour Path} is polynomial-time solvable but which only provide a yes-answer or no-answer in polynomial time. 
Note that the maximum degree of $R_k(G)$
could be equal to $(k-1)n$. This bound, together with an $O(n^2)$ bound on its diameter, only imply an 
$(kn)^{O(n^2)}$ bound on the running time of a Breadth-First Search starting in one of the colourings $\alpha$, $\beta$.

Let us discuss what is known for $1\leq k\leq 3$ or $0\leq \Delta\leq k-1$. First of all, the problem is trivial to solve if $k\leq 2$. For $k=3$, Johnson et al.~\cite{JKKPP14} proved that it is possible in 
$O(n+m)$ time to find even a shortest path between two given $k$-colourings in the reconfiguration graph $R_3(G)$ of any $3$-colourable graph~$G$ with $n$ vertices and $m$ edges.
The case $0\leq \Delta\leq k-2$ has been shown to be solvable in $O(n^2)$ time by  Cereceda~\cite{luisthesis}.
This leaves us with the case $\Delta=k-1$ and $k\geq 4$, or equivalenty, $\Delta\geq 3$ and $k= \Delta+1$.
For this case we have the following result, the proof of which can be found in Section~\ref{s-thirdproof}.

\begin{theorem}\label{t-third}
Let $G$ be a connected graph on $n$ vertices with maximum degree $\Delta\geq 3$. Let $k = \Delta+1$. If $G$ is not regular, then 
it is possible to find in $O(n^2)$ time a path between any two given $k$-colourings $\alpha$ and 
$\beta$ in $R_k(G)$.
\end{theorem}
Hence, the only remaining case, which we leave as an open problem, is when $\Delta\geq 3$, $k = \Delta+1$ and $G$ is $\Delta$-regular.  We believe that solving this case is nontrivial,  because the straightforward approach of modifying the structural proof of Theorem~\ref{t-main} does not work. As explained in Section~\ref{s-thirdproof}, such an approach would require us to find a maximum independent set for graphs of bounded maximum degree in polynomial time. However, this problem is \NP-complete even for cubic graphs~\cite{GJS76}.

\section{The Proof of Theorem~\ref{t-delta+1}}\label{s-main1}

In order to prove Theorem~\ref{t-delta+1}, we need a number of lemmas that are mostly concerned with $(\Delta+1)$-colouring. 
We define a number of terms we will use to describe vertices of $G$ with respect to some $(\Delta+1)$-colouring.  A vertex $v$ is \emph{locked} if $\Delta$ distinct colours appear on its neighbours.  A vertex that is not locked is \emph{free}.  Clearly a vertex can be recoloured only if it is free.  If $v$ is locked and then one of its  neighbour is recoloured and $v$ becomes free, we say that $v$ is \emph{unlocked}.   
A vertex $v$ is \emph{superfree} if  there is a colour $c \neq \Delta+1$ such that neither $v$ nor any of its neighbours is coloured~$c$.  A vertex can only be recoloured with a colour other than $\Delta+1$ if it is superfree.  Note there are $\Delta-1$ distinct colours that must appear on the $\Delta$ neighbours of $v$ if it is not superfree.  We say that $G$ is in \emph{$(\Delta+1)$-reduced form}  if for every vertex~$v$ coloured with $\Delta+1$, $v$ and each of its neighbours are locked.  This implies that the distance between any pair of vertices coloured $(\Delta+1)$ is at least~3 
as no locked vertex can have two neighbours coloured $(\Delta+1)$.

The key to proving Theorem~\ref{t-delta+1} will be to show that from a $(\Delta+1)$-colouring one can recolour some of the vertices to arrive at a colouring in which colour $\Delta+1$ appears on fewer vertices. We begin by considering the case where the colour $\Delta+1$ appears on only one vertex. The proof of the following lemma is inspired by a proof of Brooks' Theorem~\cite{vizing} but  also uses some new arguments.

\begin{lemma}\label{l-cubic1}
Let $G=(V,E)$ be a connected graph on $n$ vertices with maximum degree $\Delta \geq 3$, and let $\alpha$ be a $(\Delta+1)$-colouring of $G$ with exactly one vertex~$v$ coloured $\Delta+1$. If $G$ does not contain $K_{\Delta+1}$ as a subgraph, then there exists a $\Delta$-colouring $\gamma$ of $G$ such that $d_k(\alpha,\gamma)$ is $O(n)$.
\end{lemma}

\begin{proof}
We can assume that $G$ is in $(\Delta+1)$-reduced form since  if $v$ is not locked then we can immediately recolour it; if a neighbour of $v$ is not locked then it can be recoloured and this will unlock $v$ and allow us to recolour it.  

Let us fix a labelling of the neighbours of $v$: let $x_i$ be the neighbour such that $\alpha(x_i) = i$, $1 \leq i \leq \Delta$.  Our aim is to find a recolouring sequence that unlocks $v$. There is one recolouring sequence that we will use several times.  Suppose that~$C$ is a connected component of a subgraph of $G$ induced by two colours $i$ and $j$, $\Delta+1 \notin \{i,j\}$, and  no vertex coloured $j$ in $C$ is adjacent to $v$.   First the vertices coloured $j$ are recoloured with $\Delta+1$.  Then the vertices coloured $i$ are recoloured~$j$, and finally the vertices initially coloured $j$ are recoloured $i$.  It is clear that all colourings are proper and the overall effect is to \emph{swap} the colours~$i$ and $j$ on $C$.

We say that any colouring $\gamma$ where $G$ is in $(\Delta+1)$-reduced form, only $v$ is coloured $\Delta+1$ and $\gamma(x_i)=i$, $1 \leq i \leq \Delta$, is \emph{good}.  For any good colouring $\gamma$, let~$G^\gamma_{ij}$ be the maximal connected component containing $x_i$ of the subgraph of $G$ induced by the vertices coloured~$i$ and $j$ by $\gamma$.  

We make some claims about good colourings.   When we claim that $v$ can be unlocked, it is implicit that colour $\Delta+1$ is not used on any other vertex in the graph so that unlocking $v$ allows us to reach a colouring where $\Delta+1$ is not used.

\smallskip
\noindent
 \textbf{Claim 1:} If $\gamma$ is good and $x_j \notin G^\gamma_{ij}$, then $v$ can be unlocked.
\smallskip
 
\noindent If $x_j \notin G^\gamma_{ij}$, then the only vertex adjacent to $v$ in $G^\gamma_{ij}$ is $x_i$.  Thus the colours $i$ and $j$ can be swapped on $G^\gamma_{ij}$.
Then $v$ has two neighbours with colour $j$ and is unlocked.

\smallskip
\noindent
\textbf{Claim 2:} If $\gamma$ is good and $G^\gamma_{ij}$ is not a path from $x_i$ to $x_j$, then $v$ can be unlocked.

\smallskip

\noindent  By Claim 1, we can assume that $x_i$ and $x_j$ are in $G^\gamma_{ij}$.  They must have degree 1 in $G_{ij}$ since, as $G$ is in $(\Delta+1)$-reduced form, they are locked.   Suppose that $G^\gamma_{ij}$ is not a path and consider the shortest path in $G^\gamma_{ij}$ from $x_i$ to $x_j$, and the vertex~$w$ nearest to $x_i$ on the path that has degree more than 2.  Then $w$ has at least three neighbours coloured alike in $G$ and is superfree and can be recoloured with a colour other than $i$, $j$ or $\Delta+1$.  Call this new colouring $\gamma'$ and note that, by the choice of $w$, $G^{\gamma'}_{ij}$ does not contain $x_j$.  Now Claim 1 implies Claim 2.

\smallskip

As $G$  is $K_{\Delta+1}$-free, $v$ and its neighbours are not a clique so we can assume that $x_1$ and $x_2$ are not adjacent.  Let $u$ be the unique neighbour of $x_1$ coloured 2.  For a good colouring $\gamma$, note that $u$ is in $G^\gamma_{12}$, and  let $H^\gamma_{23}$ be the component of the subgraph of $G$ induced by the vertices with colour $2$ and $3$ that contains~$u$.  

\smallskip
\noindent
\textbf{Claim 3}: If $\gamma$ is good and $u$ has more than one neighbour in $H^\gamma_{23}$, then $v$ can be unlocked.

\smallskip

\noindent If $G^\gamma_{12}$ is not a path, then use Claim 2.   Otherwise $u$ has two neighbours coloured 1; if $u$ has two neighbours in $H^\gamma_{23}$, then it also has two neighbours coloured 3 and is superfree.  Recolour it and apply Claim 1.

\smallskip
\noindent
\textbf{Claim 4}: If $\gamma$ is good and $H^\gamma_{23}$ is a path, then $v$ can be unlocked.

\smallskip

\noindent By Claim 2 we can assume $G^\gamma_{23}$ is a path.  If $H^\gamma_{23}=G^\gamma_{23}$, then we can use Claim~3.  So we assume $H^\gamma_{23} \neq G^\gamma_{23}$ and so $x_2, x_3 \notin H_{23}$ and $H_{23}$ contains no neighbour of~$v$.  Let $\gamma'$ be the colouring obtained by swapping the colours 2 and 3 on $H^\gamma_{23}$.

By Claim 3, $u$ is an endvertex of $H^\gamma_{23}$.  Let the other endvertex be $w$.  (If $w=u$, then $u$ has no neighbour coloured 3 and can be recoloured.  Then use Claim 2.)

If $G^{\gamma'}_{12}$ is not a path from $x_1$ to $x_2$, we use Claim 2.  If it is such a path, then let the unique neighbour of $x_1$ in $G^{\gamma'}_{12}$ be $y$ and clearly $y \in H^\gamma_{23}$.  From $x_2$ traverse $G^{\gamma'}_{12}$ until the last vertex $z$ that is also in  $G^{\gamma}_{12}$ is reached.  Let $t$ be the next vertex along from $z$ towards $x_1$ in $G^{\gamma'}_{12}$.  Clearly $t$ is also in $H^\gamma_{23}$.  In fact, we can assume that $w=y=t$ since if $y$ or $t$ has degree 2 in $H_{23}$ as well as in $G^{\gamma'}_{12}$ it has two neighbours coloured 1 and two neighbours coloured 3 in $\gamma'$ and is superfree.  It can be recoloured and then Claim 2 is used.  

So $x_1wz$ is coloured $131$ in $\gamma$ so is in $G^\gamma_{13}$.  Then $z$ is in both $G^\gamma_{13}$ and $G^\gamma_{12}$ so is superfree and can be recoloured so that Claim 2 can be used.  This completes the proof of Claim~4.

\smallskip

To complete the proof: we know that the initial colouring $\alpha$ is good. If none of the four claims can be used, then consider $H^\alpha_{23}$.  We know that $u$ has degree 1 in $H_{23}$ but $H_{23}$ is not a path.  So traversing edges away from $u$ in $H^\alpha_{23}$, let $s$ be the first vertex reached  with degree 3.  Then $s$ is superfree and can be recoloured so that $H_{23}$ becomes a path, and then Claim  4 can be used.
\qed
\end{proof}

In Lemma~\ref{l-cubic2}, we shall see how 
the number of vertices coloured $\Delta+1$ can be reduced when more than one is present.  First we need some definitions and a lemma. Let $P$ be a path: 
  \begin{itemize}
   \item  $P$ is \emph{nearly $(\Delta+1)$-locked} if its endvertices are locked and coloured $\Delta+1$;
   \item $P$ is \emph{$(\Delta+1)$-locked} if it is nearly $(\Delta+1)$-locked and every vertex on the path is locked.  
  \end{itemize}

\begin{lemma} \label{l-first}
Let $G$ be a graph in $(\Delta+1)$-reduced form. If $G$ has a $(\Delta+1)$-locked path $P$, then each endvertex of $P$ is an endvertex of an $(\Delta+1)$-locked path of length 3.
\end{lemma}

\begin{proof}
As we noted, by the definition of $(\Delta+1)$-reduced form, a path between two vertices coloured $\Delta+1$ has length at least 3.   Let $u$ be one endvertex of $P$ and let $Q$ be the shortest $(\Delta+1)$-locked path that ends at $u$ (so $Q$ is induced).  Let $v$ be the vertex on $Q$ at distance 2 from $u$.  Then, as $v$ is locked and not a neighbour of $u$, it has a neighbour $w$ coloured $\Delta+1$ that is not $u$ and the path from $u$ to $w$ has length 3.
\qed
\end{proof}

A path is \emph{nice} if it is a nearly $(\Delta+1)$-locked path, it contains free vertices and the  endvertices and their neighbours are the only locked vertices. Notice that a nice path is not necessarily induced and, in particular, may contain a $(\Delta+1)$-locked subpath.

\begin{lemma}\label{l-cubic2}
Let $G$ be a connected 
 graph on $n$ vertices with  maximum degree $\Delta \geq 3$, let $\alpha$ be a $(\Delta+1)$-colouring of $G$, and suppose that  $G$ is in $(\Delta+1)$-reduced form. If $G$ has at least two $(\Delta + 1)$-locked vertices and is not frozen, then there exists a $(\Delta+1)$-colouring $\gamma$ of $G$, such that $d_{\Delta+1}(\alpha, \gamma) =  O(n)$ and fewer vertices are coloured $\Delta+1$ with $\gamma$ than with $\alpha$.
\end{lemma}

\begin{proof} We consider a number of cases.

\smallskip
\textbf{Case 1}: There exists a free vertex $u$ adjacent to a $(\Delta+1)$-locked path $P$. 
\smallskip

\noindent Let $b$ be the vertex on the path adjacent to $u$. 
As $b$ is locked it has a neighbour $a$ coloured $\Delta+1$.  Let $c$ be a neighbour of $b$ on $P$ other than $a$.  As $c$ is locked it has a neighbour $d$ coloured $\Delta+1$.  

Since $G$ is in $(\Delta+1)$-reduced form, $u$ is not adjacent to $a$ or $d$ but might be adjacent to $c$.  In each case, it is routine to verify that by recolouring $u$ to $\Delta+1$,  $b$ and $c$ can both be recoloured unlocking $a$ and $d$ and allowing them to be recoloured.  Thus the number of vertices coloured $\Delta +1$ is reduced.

\smallskip
\textbf{Case 2}:  $G$ has a nice path. 
\smallskip

\noindent  Let $P$ be a shortest nice path.  Let the endpoints  be $v$ and $w$ with neighbours $s$ and $t$ on $P$ respectively.  If $s$ and $t$ are adjacent, then the path $vstw$ is $(\Delta+1)$-locked and has a free vertex adjacent to $s$ so use Case 1.  Thus assume that $P$ is induced since the presence of any other edge would imply either a shorter nice path could be found or that the graph was not in $(\Delta+1)$-reduced form.

We use induction on the number $\ell$  of free vertices in $P$ to show that there is a sequence of recolourings that lead to a colouring that has fewer vertices coloured~$\Delta+1$. 

If $\ell = 1$, let $u$ be the free vertex in $P$. Recolour $u$ to $\Delta+1$. Now~$s$ and $t$ have two neighbours coloured $\Delta+1$ and can be recoloured. Then~$v$ and~$w$ are unlocked and can both be recoloured, and this leaves one vertex on~$P$ coloured $\Delta+1$ rather than two.

Suppose that $\ell = 2$. Let $P = vsu_1u_2tw$ where $u_1$ and $u_2$ are  free vertices.  First suppose that $u_1$ or $u_2$, say $u_1$, is superfree: recolouring $u_1$ to a colour $c \not= \Delta+1$ unlocks $s$; recoloured $s$ unlocks $v$ which, in turn, allows us to recolour it, and the number of vertices coloured $\Delta+1$ has been reduced as required. 
Similarly if $u_1$ is not superfree but a neighbour $x$ is, then $x$ can be recoloured to a colour $c \neq \Delta+1$ and if $xs$ is an edge then $s$ is unlocked and so $v$ can be recoloured, and otherwise $u$ is now superfree, the colours in the neighbourhood of $s$ are unchanged, and the preceding argument can be applied.

Thus henceforth we can assume that $u_1$, $u_2$ and their neighbours are not superfree which implies that they have degree $\Delta$.

\smallskip
\noindent
\textbf{Subcase 2.1}: \textit{$u_1$ and $u_2$ do not share a neighbour}.
Let $x_1$ and $x_2$ be neighbours of $u_1$ and $u_2$ not in~$P$.  Clearly $x_1 \not= x_2$ and $u_1x_2$ and $u_2x_1$ are not edges.

\smallskip
\noindent
\textbf{Subcase 2.1.1}: \textit{$x_1$ is locked}. 
We know $x_1$ has a $(\Delta+1)$-locked neighbour, and this must be $v$ (if it is some other vertex $z$, then $vsu_1x_1z$ is a nice path that is shorter than $P$).

Suppose $x_1s$ is not an edge. Recolour $u_1$ to $\Delta+1$.  This unlocks $x_1$ which can be recoloured with $\alpha(u_1)$ which, in turn, unlocks $v$ and allows us to recolour it with $\alpha(x_1)$.  If $u_1$ is free, it can be recoloured and the number of vertices coloured $\Delta+1$ is reduced and we are done.  If $u_1$ is locked, then note that $s$ has been unlocked (as it no longer has a neighbour coloured $\alpha(u_1)$).  Thus we can recolour~$s$ and then recolour $u_1$ with $\alpha(s)$ and again we have removed one instance of the colour $\Delta+1$.

Suppose instead that $x_1s$ is an edge.   Notice that  $\alpha(s)$, $\alpha(u_1)$ and $\alpha(x_1)$ are distinct as the three vertices form a triangle. Recolour $u_1$ with $\Delta+1$ and then~$s$ with $\alpha(u_1)$. Now $v$ is unlocked and can be recoloured with $\alpha(s)$.  If $u_1$ is free, then recolour it and we are done.  Otherwise this sequence of recolourings leaves $u_1$ locked (with $\alpha(u_1)$ and $\alpha(x_1)$ as the colours on $s$ and $x_1$ respectively).  So, from $\alpha$, we do the following instead: again start by recolouring $u_1$ with $\Delta+1$, but then recolour $x_1$ with $\alpha(u_1)$ to unlock $v$.  Now that $\alpha(x_1)$ is not used on a neighbour of $u_1$, $u_1$ is free and can be recoloured.

\smallskip
\noindent
\textbf{Subcase 2.1.2}: \textit{$x_1$ is free}. 
If $x_2$ is locked, we can, by symmetry, use the previous subcase, so we can assume that both $x_1$ and $x_2$ are free. Recolour $u_2$ to $\Delta+1$.  Then $t$  is unlocked and can be recoloured which, in turn, unlocks $w$ allowing us to recolour it too. If $u_2$ is free, we recolour it and are done.  If $u_1$ is free, we recolour it and unlock $u_2$ and, again, recolour it.

If $u_1$ and $u_2$ are both locked, observe that $x_1$ is still free as it has no neighbour coloured $\Delta+1$ since $u_2x_1$ is not an edge. Recolour $x_1$ to $\Delta+1$, and then recolour $u_1$ to $\alpha(x_1)$.  Note that now $s$ has no neighbour coloured $\alpha(u_1)$ and is free and can be recoloured so that $v$ is unlocked and can also be recoloured.  By recolouring $u_1$, we also unlock $u_2$, so we recolour it and are done.

\smallskip
\noindent
\textbf{Subcase 2.2}:  \textit{$u_1$ and $u_2$ share a neighbour}.
Let $x_1$ be a neighbour of $u_1$ and $u_2$. Since $P$ is induced, $x_1$ is not in $P$.  If $x_1$ is locked, then let its neighbour coloured $\Delta+1$ be $y$.  Then $vsu_1x_1y$ is a shorter nice path unless $y=v$.  By an analagous argument we need $y=w$.  This contradiction tells us that $x_1$ must be free.

If $x_1$ is joined to both $s$ and $t$, then $vsx_1tw$ is a shorter nice  path.  So, without loss of generality, assume that $x_1t$	is not an edge. Thus as $u_2$ has a neighbour that is not adjacent to $x_1$, $x_1$ has a neighbour $x_3$ that is not adjacent to $u_2$ since both have degree $\Delta$.

\smallskip
\noindent
\textbf{Subcase 2.2.1}: \textit{$x_3 = s$}.
Recolour $u_1$ with $\Delta+1$ and then $s$ with $\alpha(u_1)$. Now $v$ is unlocked and can be recoloured with $\alpha(s)$. If $u_1$ is free, then recolour it and we are done. If $u_2$ or $x_1$ is still free, then recolour one of them to unlock $u_1$, which in turn can be recoloured and are done. Otherwise this sequence of recolourings leaves $u_1, u_2$ and $x_1$ locked so $x_1$ is the only neighbour of $u_2$ coloured $\alpha(x_1)$.  So, from $\alpha$, we do the following instead: recolour $x_1$ with $\Delta+1$ to unlock $s$ and then $v$. If $x_1$ can be recoloured, then we do so and are done. Otherwise notice that $\alpha(x_1)$ is not used on a neighbour of $u_2$. It is thus free and can be recoloured to unlock $x_1$ and allow us to recolour it.

\smallskip
\noindent
\textbf{Subcase 2.2.2}:  \textit{$x_3 \not=s$, and $x_3$ is free}.
First, suppose $x_3s$ is an edge.  Recolour $u_2$ to $\Delta+1$,  $t$ to $\alpha(u_2)$ and $w$ to $\alpha(t)$. If either $u_2$ or one of its neighbours is now free, $u_2$ can be recoloured and we are done. Otherwise $u_1$, $u_2$ and $x_1$ are all locked,  but $x_3$ is still free since it has no neighbour coloured $\Delta+1$. Recolour $x_3$ to $\Delta+1$ to unlock $x_1$; then recolour $x_1$ to unlock and recolour $u_2$.  As $x_3s$ is an edge, $s$ has two neighbours coloured $\Delta+1$. Thus we recolour $s$ to unlock $v$.

If $x_3t$ is an edge we can use a similar argument.  So suppose $x_3s$ and $x_3t$ are not edges. Recolour $u_2$ to $\Delta+1$, to unlock and recolour first $t$ and then $w$.  It is possible to recolour $u_2$ unless it and all its neighbours are locked. This implies that $u_1$, $x_1$ and $u_2$ are locked. We consider two subcases.

\smallskip
\noindent
\textbf{Subcase 2.2.2.1}: \textit{$u_1x_3$ is not an edge}.
We recolour $x_3$ to $\Delta+1$ to unlock and recolour $x_1$ and then $u_2$. Notice that $u_1$ is now free since it has no neighbour coloured $\Delta+1$. Recoloured $u_1$ unlocks $s$, so we recolour it, which in turn unlocks~$v$. Observe that $x_1$ now has two neighbours $u_1$ and $x_3$ with colour $\Delta+1$ so is free.  If $u_1$ or $u_3$ is free, we can recolour at least one of them directly and we are done. Otherwise, we recolour $x_1$ so that $x_3$ and $u_1$ can now be recoloured.

\smallskip
\noindent
\textbf{Subcase 2.2.2.2}: \textit{$u_1x_3$ is an edge}.
Recolour $u_3$ to $\Delta+1$, then recolour $u_1$, $s$ and $v$. Observe that $x_1$ now has two neighbours~$u_2$ and $u_3$ with colour $\Delta+1$. If $u_2$ or $u_3$ are free, we are done. Otherwise, recolour $x_1$, then recolour $u_2$ and $x_3$, and we are done. 

\smallskip
\noindent
\textbf{Subcase 2.2.3}:  \textit{$x_3\not=s$, and $x_3$ is locked}.
Then $x_3$ has a $(\Delta+1)$-locked neighbour~$y$. If $y = v$, the path $H =  vx_3x_1u_2tw$ is nice with two free vertices $x_1$ and $u_2$. Furthermore, $u_1$ is free and a neighbour of $x_1$ and $u_2$, in which case $H$ satisfies the previous subcase unless $x_3$ and $t$ are adjacent in which case use Subcase 2.1. A similar argument can be made if $y = w$ or $y \not\in \{v,w\}$. 

This completes the case $\ell = 2$.  (We note that if we wished to use the proof to construct an algorithm, we would first check whether $x_3$ is superfree as in this case the proof can be simplified in many places.)

\smallskip

Now suppose that for all $i < \ell$, if there is a nice path containing $i$ free vertices, the number of vertices coloured $\Delta+1$  can be reduced.  Suppose that the shortest such path is $P=vsu_1u_2 \dots u_\ell tw$ where $\ell \geq 3$. We recolour $u_\ell$ to $\Delta+1$, then $t$ and then $w$.  If $u_\ell$ or one of its neighbours is free, then $u_\ell$ can be recoloured and we are done. Otherwise, $u_\ell$ and $u_{\ell-1}$ are locked. Consider the path $P' = vsu_1 \dots u_{\ell-2}u_{\ell-1}u_\ell$. By our inductive hypothesis, the number of colour $\Delta+1$ vertices in $P'$ can be reduced.  Case 2 is complete.

\smallskip
\noindent  After Cases 1 and 2 we are left with:
\smallskip

\textbf{Case 3}: There does not exist a free vertex adjacent to a $(\Delta+1)$-locked path and $G$ has no nice path.

\smallskip

\noindent As $G$ contains more than one $(\Delta+1)$-locked vertex, it contains a nearly $(\Delta+1)$-locked path; let $P$ be the shortest and let $v$ and $w$ be its endvertices.  As $G$ is in $(\Delta+1)$-reduced form, $v$, $w$ and their neighbours are locked.  If $P$ contains no other vertices, it is $(\Delta+1)$-locked.  Otherwise, since there are no nice paths, $P$ contains another locked vertex $u$.  Let $y$ be the neighbour of $u$ coloured $\Delta+1$.   If~$y$ is on $P$, then we can assume, without loss of generality, that it is not between~$v$ and $u$.  Then, whether or not $y$ is on $P$, the subpath from $v$ to $u$ plus the edge $uy$ is a shorter nearly $(\Delta+1)$-locked path.  This contradiction proves that $G$ must contain a $(\Delta+1)$-locked path.

As $G$ is not frozen, it contains a free vertex.  Let $Q$ be the shortest path in $G$ that joins a free vertex to a  $(\Delta+1)$-locked vertex.  Let $v$ be the $(\Delta+1)$-locked endvertex.  So $v$ is an endpoint of a $(\Delta+1)$-locked path~$R$, and, by Lemma~\ref{l-first}, we can assume that $R$ has length $3$. 

Let $u$ be the endvertex of $Q$ that is free. By the minimality of $Q$, $u$ is the only free vertex in $Q$.  Let $a$ be the neighbour of $u$ in $Q$.  As $a$ is locked it has a $(\Delta+1)$-locked neighbour $z$.   Thus  we must have $z=v$ and $Q=vau$.

Let $R = wtsv$.  Observe that $us$, $ut$, $uv$ and $uw$ cannot be edges as no locked path has a free neighbour. Thus the vertices of $R$ and $Q$ other than $v$ are distinct.  Consider the (not necessarily induced) path $M= wtsvau$.  Notice also that $at$ is not an edge else the free vertex $u$ is adjacent to the $(\Delta+1)$-locked path $vatw$.

Suppose $M$ is an induced path. Recolour $u$ with $\Delta+1$ to unlock and recolour $a$ and then $v$.  If $u$ is not locked, then recolour and we are done. Else notice that the vertices $v$ and $s$ are free, and the vertices $u,a, t, w$ are locked. Consequently, we have that $M$ is a nice path, and by Case $2$ we are done. 

The only edge that might be present among the vertices of $M$ is $as$ so suppose this exists. Recolour $u$ with $\Delta+1$ to unlock and recolour first $a$ and then $v$. If $u$ or any of its neighbours are free, $u$ can be recoloured and we are done.  Otherwise note that recoloured $v$ unlocks $s$. It follows that the path $H = uastw$ is nice, and we can use Case $2$. This completes Case 3. 

\smallskip

\noindent As each vertex is recoloured a constant number of times,  the lemma follows. 
\qed
\end{proof}

We are now ready to prove Theorem~\ref{t-delta+1}, which we first restate.

\medskip
\noindent
{\bf Theorem~\ref{t-delta+1}.}
{\it 
Let $G$ be a connected graph on $n$ vertices with maximum degree~$\Delta \geq 1$, and let $k \geq \Delta+1$.    Let $\alpha$ be a $k$-colouring of $G$.  If $\alpha$ is not frozen and $G$ is not $K_{\Delta+1}$ or, if $n$ is odd, $C_n$, then there exists a $\Delta$-colouring $\gamma$ of $G$ such that~$d_k(\alpha, \gamma)$ is $O(n^2)$.}

\begin{proof}
If $k > \Delta +1$, then, by Brooks' Theorem, a $\Delta$-colouring $\gamma$ exists in $R_{k}(G)$ unless $G$ is complete or an odd cycle.  We know that, in this case, $R_{k}(G)$ is connected and has diameter $O(n^2)$ so certainly $d_k(\alpha, \gamma)$ is $O(n^2)$.  

Suppose that $k = \Delta +1$ and $G$ is in $(\Delta+1)$-reduced form with $\alpha$: if not, we try to recolour each vertex with colour $\Delta+1$ either directly or by first recolouring one of its neighbours. Repeatedly applying Lemma \ref{l-cubic2} starting from $\alpha$, we obtain a $(\Delta+1)$-colouring $\gamma'$ by $O(n^2)$ recolourings such that at most one vertex is coloured $(\Delta+1)$ with $\gamma'$. Lemma $\ref{l-cubic1}$ can now be applied to obtain a $\Delta$-colouring $\gamma$ from $\gamma'$ by $O(n)$ recolourings.  Hence, $d_{\Delta+1}(\alpha, \gamma) \leq O(n^2)$ as required.  \qed
\end{proof}

\section{The Proof of Theorem~\ref{t-main}}\label{s-main2}

First we need the following result of Matamala~\cite{matamala}. We use $\omega(G)$ to denote the number of vertices in the largest clique in~$G$. 

\begin{lemma}[\cite{matamala}]\label{thm3}
Let $G=(V,E)$ be a graph with maximum degree $\Delta \geq 3$ and $\omega(G) \leq \Delta$.   Let $p_1$ and $p_2$ be non-negative integers such that $p_1+p_2=\Delta-2$.   Then there is a partition $\{S_1,S_2\}$ of $V$  such that $S_1$ induces a maximum size $p_1$-degenerate graph in  $G$ and $S_2$ induces a $p_2$-degenerate graph.
\end{lemma}

We also need the following two lemmas.

\begin{lemma}\label{l-deg}
Let $G$ be a connected $(\Delta-1)$-degenerate graph on $n$ vertices with maximum degree~$\Delta \geq 3$, and let $k \geq \Delta+1$. Let $\alpha$ be a $k$-colouring of $G$.  Then there exists a $\Delta$-colouring $\gamma$ of $G$ such that~$d_k(\alpha, \gamma)$ is $O(n^2)$.  
\end{lemma}

\begin{proof}
The result follows immediately from Theorem~\ref{t-delta+1} by observing that a $(\Delta - 1)$-degenerate graph has a vertex with at most $\Delta - 1$ neighbours, 
so $\alpha$ is not frozen and $G$ is not $K_{\Delta+1}$ or $C_n$.
\qed
\end{proof}

\begin{lemma} \label{l-delta}
Let $G=(V,E)$ be a graph on $n$ vertices with maximum degree~$\Delta \geq~1$.  Let $\gamma_1$ and $\gamma_2$ be $\Delta$-colourings of $G$.  Then $d_{\Delta+1}(\gamma_1, \gamma_2)$ is $O(n^{2})$. 
\end{lemma}  

\begin{proof}
We use induction on $\Delta$.
If $\Delta \in \{1, 2\}$ the statement is trivially true.  Let $\Delta\geq 3$.
We observe that $\omega(G) \leq \Delta$ because $G$ is $\Delta$-colourable. Applying Lemma~\ref{thm3} with $p_1=0$ and $p_2=\Delta-2$, we obtain a partition $\{S_1, S_2\}$ of $V$ such that $S_1$ is a maximum independent set and $S_2$  induces a $(\Delta-2)$-degenerate graph that we denote by~$H$. 

From $\gamma_1$ and $\gamma_2$ recolour the vertices of $S_1$ with colour $\Delta+1$ (the colour that is not used in either $\gamma_1$ or $\gamma_2$). This can be done by at most $2n$ recolourings.
So now we can focus on the colourings restricted to~$S_2$, and as long as we do not use the colour $\Delta+1$ we do not need to worry about adjacencies with $S_1$.  So let $\gamma_1^H$ and $\gamma_2^H$ be the colourings of $H$ that are obtained by taking the restrictions of $\gamma_1$ and $\gamma_2$ to $S_2$, and
if we can recolour from $\gamma_1^H$ to $\gamma_2^H$ by $O(n^2)$ recolourings
without using colour $\Delta+1$, we will be done.  We note that 
\begin{itemize}
\item $\gamma_1^H$ and $\gamma_2^H$ use only colours from $\{1, 2, \ldots, \Delta\}$;
\item each component of $H$ has maximum degree at most $\Delta - 1$ (since every vertex in 	$S_2$ has at least one neighbour in $S_1$ by the maximality of $S_1$);
\item each component of $H$ is $(\Delta-2)$-degenerate.
\end{itemize}
Thus we can apply Lemma~\ref{l-deg} on each component of $H$
to recolour each of $\gamma_1^H$ and $\gamma_2^H$ to a $(\Delta-1)$-colouring
using at most $O(n^2)$ recolourings. By the inductive hypothesis, there is a path 
of length $O(n^2)$ between these two $(\Delta-1)$-colourings that includes only $\Delta$-colourings so does not use colour $\Delta+1$. Because at most $2n$ recolourings were needed to recolour $\gamma_1$ and $\gamma_2$ to $\gamma_1^H$ and $\gamma_2^H$, the total number of recolourings is $O(n^2)$. This completes the proof of Lemma~\ref{l-delta}.
\qed
\end{proof}

The lemma says that there is a path between any pair of $\Delta$-colourings, but, because we are working with $R_{\Delta+1}(G)$, the intermediate colourings might use $\Delta+1$ colours. We are now ready to prove Theorem~\ref{t-main}, which we restate below.

\medskip
\noindent
{\bf Theorem~\ref{t-main}.}
{\it Let $G$ be a connected graph on $n$ vertices  with maximum degree~$\Delta \geq 3$.    Let $\alpha$ and $\beta$ be $(\Delta+1)$-colourings of $G$.  If  $\alpha$ and $\beta$ are not frozen colourings, then $d_{\Delta+1}(\alpha, \beta)$ is $O(n^2)$.}

\begin{proof}
Theorem~\ref{t-delta+1} implies that from each of $\alpha$ and $\beta$ there is a path in $R_{\Delta+1}$ to a $\Delta$-colouring; Lemma~\ref{l-delta} implies that there is a path between these two $\Delta$-colourings that completes the path from $\alpha$ to $\beta$. 
\qed
\end{proof}

\section{The Proof of Theorem~\ref{t-third}}\label{s-thirdproof}

The \emph{degeneracy} of a graph $G$ is the least integer $k$ such that $G$ is $k$-degenerate.

We start with the following easy lemma, which is well known (see, for example,~\cite{matamala}). We give a short proof for completeness.

\begin{lemma}\label{l-only}
Let $\Delta\geq 1$.
Every connected graph with maximum degree~$\Delta$ that is not regular is $(\Delta-1)$-degenerate.
\end{lemma}

\begin{proof} 
Let $G$ be a smallest possible counterexample, so $G$ has degeneracy and maximum degree equal to $\Delta$ and contains a vertex $v$ with  $\dgr(v) < \Delta$.  If $G-v$ has degeneracy $\Delta$ then, by the minimality of $G$, we find that $G-v$ is $\Delta$-regular. This means that in $G$ every neighbor of $v$ has more than $\Delta$ neighbours, which is not possible. Hence, $G-v$ has degeneracy $\Delta-1$. Every induced subgraph $G'$ of $G$ is either an induced subgraph of $G-v$ or contains $v$.
Hence, $G'$ has a vertex of degree less than $\Delta$ contradicting the claim that $G$ has degeneracy~$\Delta$.\qed 
\end{proof}

Lemma~\ref{l-only} 
tells us that Theorem~\ref{t-third} is a statement about
$(\Delta-1)$-degenerate graphs. 

We introduce some additional definitions.
We let $G[S]$ denote the subgraph of a graph $G=(V,E)$ induced by some set $S\subseteq V$.
It is well-known that $G$ is $p$-degenerate for some integer~$p$ if and only if there exists a \emph{degeneracy ordering} $v_1, v_2, \dots, v_n$ of its vertices such that $v_i$ has at most $p$ neighbours $v_j$ with $j < i$.  One can compute such an ordering
in  $O(n^2)$ time (let $v_n$ be a vertex of minimum degree in~$G$ and, for $i=n-1,\ldots, 1$, let $v_i$ be a vertex of minimum degree in $G[V\setminus \{v_{i+1},\ldots,v_n\}]$).

We need an algorithmic version of a result of Mih\'ok~\cite{mihok}, which was proven independently by Wood~\cite{wood}. 
We present a slightly modified version of the proof of Wood which 
was implicitly algorithmic (it suffices to make a few additional algorithmic observations).
  
\begin{lemma}[\cite{mihok,wood}]\label{degpartition}
Let $r\geq 1$ and $k\geq r-1$.
Let $G = (V, E)$ be a $k$-degenerate graph on $n$ vertices. Let $p_1, \dots, p_r$ be non-negative integers so 
that $\sum_{t=1}^r{p_t} = k - r + 1$.
Then it is possible to compute in $O(n^2)$ time a partition $\{V_1, \dots, V_r\}$ of $V$ such that $G[V_t]$ is $p_t$-degenerate
for $t = 1, \dots, r$.
\end{lemma}

\begin{proof}
We first compute a degeneracy ordering $v_1,  \dots, v_n$ of $G$ in $O(n^2)$ time.
For $i=1,\ldots,n$, we define $X_i = \{v_1, \dots, v_i\}$. Then, by definition, $v_i$ has at most $k$ neighbours in 
$X_{i-1}$. It suffices to prove the invariant that we can compute in $O(n)$ time 
a partition $\{Y_1, \dots, Y_r\}$ of $X_i$, where $G[Y_{s}]$ is $p_s$-degenerate for $s = 1, \dots, r$. 
If $i = 1$ the invariant trivially holds. Let $i\geq 2$.
By our invariant, we can compute in $O(n)$ time a partition $\{Z_1, \dots, Z_r\}$ of $X_{i-1}$ where $G[Z_s]$ is $p_s$-degenerate
for $s = 1, \dots, r$.  If $v_{i}$ has more than $p_s$ neighbours in every $G[Z_{s}]$ then $v_i$ has at 
least $\sum_{i=1}^r(p_i+1) = k + 1$ neighbours in 
$X_{i-1}$, a contradiction. Hence,
 $v_i$ has at most $p_q$ neighbours in at least one set $Z_q$, which we can find in $O(n)$ time.
We put $v_i$ into $Z_q$ to get the desired partition for $X_i$ in $O(n)$ time.  \qed
\end{proof}

Recall that Cereceda~\cite{luisthesis} proved that for any $k \geq 2d + 1$ 
the reconfiguration graph $R_k(G)$ of every $d$-degenerate graph $G$ on $n$ vertices has diameter $O(n^2)$. We adapt his proof to show the following lemma.

\begin{lemma}\label{l-degenerate}
Let $G$
be a graph on $n$ vertices with maximum degree $\Delta \geq 1$ and degeneracy $\Delta - 1$.  Let $\alpha$ be a $(\Delta+1)$-colouring of~$G$. It is possible to compute a $\Delta$-colouring $\gamma$ of $G$ in time $O(n^2)$ such that
\mbox{$d_{\Delta+1}(\alpha, \gamma) \leq n^2$}. 
\end{lemma}

\begin{proof}
We first compute a degeneracy ordering $v_1,  \dots, v_n$ of $G$ in $O(n^2)$ time.  Also in $O(n^2)$ time we record, for each vertex $v$, the neighbour of $v$ that is latest in the ordering, and the set of colours that are not used on neighbours of~$v$.

Let $h$ be the lowest index such that $\alpha(v_h) = \Delta + 1$.  We will describe an algorithm that finds in time $O(n)$ a sequence of recolourings such that 
\begin{itemize}
\item for $i<h$, $v_i$ is not recoloured,
\item for $i\geq h$, $v_i$ is recoloured at most once, and
\item $v_h$ is recoloured with a colour other than $\Delta+1$.
\end{itemize}
By repeatedly using such sequences, we can obtain a colouring $\gamma$ in which colour $\Delta+1$ is not used. At most $n$ such sequences are needed, so each vertex is recoloured at most $n$ times and the lemma follows.

We must describe the algorithm.  First we find a sequence $S$ of pairs of vertices and colours $(w_j, c_j)$ as follows: 
\begin{itemize}
\item the first vertex $w_1$ is $v_h$;
\item for each vertex $w_j$, if there is a colour that is not used on it or any of its neighbours then this is $c_j$, and $(w_j, c_j)$ is the final pair in $S$;
\item otherwise let $w_{j+1}$ be the neighbour of $w_j$ that is latest in the degeneracy ordering and let $c_j=\alpha(w_{j+1})$.
\end{itemize}

\noindent If all $\Delta+1$ colours  appear on $w_j$ and its neighbours, then $w_j$ must have degree~$\Delta$, each neighbour of $w_j$ must have a distinct colour, and, as at most $\Delta-1$ neighbours can be earlier in the degeneracy ordering, at least one neighbour is later in the ordering.  Thus each vertex in $S$ is later in the  degeneracy ordering than its predecessor and so the algorithm will terminate and $S$ is finite. Moreover, this also implies that each vertex in $v_{h+1},\ldots,v_n$ is considered at most once during the computation of $S$ and so, as the information required about each vertex was found during our preliminary computations, we can find $S$ in $O(n)$ time.

Let $s$ denote the number of pairs in $S$. We can recolour the vertices of $S$ in time $O(n)$ by simply recolouring $w_j$ with~$c_j$, starting with $w_s$ and working backwards through $S$.  Each colouring obtained is proper since $w_s$ has no neighbour coloured $c_s$ and when a vertex~$v_j$, $j<s$ is recoloured, its unique neighbour $w_{j+1}$ coloured $c_j$ has just been recoloured and it is not adjacent to any other vertex that has been recoloured since they are later in the degeneracy ordering than any of its neighbours.  Finally note that $w_1=v_h$ has been recoloured with a colour other than $\alpha(v_h) = \Delta+1$, so the recolouring sequence achieves its aim. This completes the proof.
 \qed 
 \end{proof}
  
Finally we need an algorithmic version of Lemma~\ref{l-delta} for the special case of $(\Delta-1)$-degenerate graphs; to prove it we follow the line of the proof
of Lemma~\ref{l-delta} but need Lemma~\ref{degpartition} instead of Lemma~\ref{thm3} and Lemma~\ref{l-degenerate} instead of Lemma~\ref{l-deg}.
The question whether there exists an algorithmic version for the remaining case of $\Delta$-regular graphs is still open;
note that one cannot replace Lemma~\ref{degpartition} by Lemma~\ref{thm3} in the proof of Lemma~\ref{l-delta2}, as that would require solving the \NP-compete problem of finding a maximum independent set in graphs of bounded maximum degree in polynomial time. 

\begin{lemma} \label{l-delta2}
Let $G=(V,E)$ be a $(\Delta-1)$-degenerate graph on $n$ vertices with maximum degree $\Delta \geq 1$.  
It is possible to find in $O(n^2)$ time a path between any two given $\Delta$-colourings $\gamma_1$ and $\gamma_2$ in $R_{\Delta+1}(G)$. 
\end{lemma}  

\begin{proof}
We use induction on $\Delta$.
If $\Delta \in \{1, 2\}$ the statement is trivially true.  Let $\Delta\geq 3$ and assume that we have an $O(n^2)$-time algorithm for 
connected $(\Delta-2)$-degenerate graphs on $n$ vertices with maximum degree $\Delta-1$.

Applying Lemma~\ref{degpartition} with $p_1 = 0$ and $p_2 = \Delta - 2$
gives us in $O(n^2)$ time a partition $\{S_1, S_2\}$ of $V$ such that $S_1$ is an independent set and $S_2$  induces a $(\Delta-2)$-degenerate graph
that we denote by $H$.
We modify the pair $(S_1,S_2)$ in $O(n^2)$ time by moving vertices from $S_2$ to $S_1$ until $S_1$ is a maximal independent set.
Let $\gamma_1^H$ and $\gamma_2^H$ be the colourings of $H$ that are the restrictions of $\gamma_1$ and $\gamma_2$ to $S_2$.
We note that 
\begin{itemize}
\item $\gamma_1^H$ and $\gamma_2^H$ use only colours from $\{1, 2, \ldots, \Delta\}$;
\item $H$ has maximum degree at most $\Delta - 1$ (by the maximality of $S_1$);
\item $H$ is $(\Delta-2)$-degenerate.
\end{itemize}
Thus we can apply Lemma~\ref{l-degenerate} to recolour each of $\gamma_1^H$ and $\gamma_2^H$ to a $(\Delta-1)$-colouring
in $O(n^2)$ time.
We then apply the induction hypothesis to find in $O(n^2)$ time  
a path between these two $(\Delta-1)$-colourings that includes only $\Delta$-colourings.
Hence the total running time is $O(n^2)$, as required.\qed
\end{proof}

We are now ready to prove Theorem~\ref{t-third}, which we first restate.

\medskip
\noindent
{\bf Theorem~\ref{t-third}.}
{\it Let $G$ be a connected graph on $n$ vertices with maximum degree $\Delta\geq 3$. Let $k = \Delta+1$. If $G$ is not regular, then 
it is possible to find in $O(n^2)$ time a path between any two given $k$-colourings $\alpha$ and 
$\beta$ in $R_k(G)$.
}

\begin{proof}
By Lemma~\ref{l-only} we find that $G$ is $(\Delta-1)$-degenerate.
By Lemma~\ref{l-degenerate} we can find in $O(n^2)$ time a path from $\alpha$ to some $\Delta$-colouring $\gamma_1$ and a path from $\beta$ to some $\Delta$-colouring $\gamma_2$. Applying Lemma~\ref{l-delta2} completes the proof.\qed
\end{proof}

\bibliography{bibliography}{}
\bibliographystyle{abbrv}

\end{document}